\documentclass[english,aps,pra,twocolumn]{revtex4-1}
\usepackage[T1]{fontenc}
\usepackage[latin9]{inputenc}
\usepackage{mathrsfs}
\usepackage{amsthm}
\usepackage{amsmath}
\usepackage{amssymb}
\usepackage{graphicx}
\usepackage{esint}

\theoremstyle{plain}
\newtheorem{thm}{\protect\theoremname}
  \theoremstyle{plain}
  \newtheorem{lem}[thm]{\protect\lemmaname}
  \theoremstyle{plain}
  \newtheorem{cor}[thm]{\protect\corollaryname}

\usepackage{babel}
  \providecommand{\corollaryname}{Corollary}
  \providecommand{\lemmaname}{Lemma}
\providecommand{\theoremname}{Theorem}

\begin{document}

\title{Average position in quantum walks with a U(2) coin}

\author{Min Li}

\author{Yong-Sheng Zhang}

\email{Email: yshzhang@ustc.edu.cn}

\author{Guang-Can Guo}

\affiliation{Key Laboratory of Quantum Information, University of Science and
Technology of China, CAS, Hefei, 230026, People's Republic of China}

\date{\today}
\begin{abstract}
We investigated discrete-time quantum walks with an arbitary unitary
coin. Here we discover that the average position $\left\langle x\right\rangle =\max(\left\langle x\right\rangle )\sin(\alpha+\gamma)$,
while the initial state is $1/\sqrt{2}(\mid0L\rangle+i\mid0R\rangle)$.
We prove the result and get some symmetry properties of quantum walks
with a U(2) coin with $\mid0L\rangle$ and $\mid0R\rangle$ as the
initial state.
\end{abstract}
\maketitle

\section{introduction}

Quantum walks (QWs) were first introduced in 1993 \cite{Aharonov}
as generalization of classical random walks. According to the time
evolution, QWs can be devided into discrete-time and continuous-time
\cite{Farhi-continue-time} QWs. Recently, both continuous-time \cite{Childs-continue-universal}
and discrete-time \cite{Lovett-discrete-universal} QWs are found
to be universal for quantum computation. A number of quantum algorithms
based on QWs have already been proposed in \cite{Shenvi2003,Childs2004,Ambainis,Ambainis2004,Childs2003,Childs2002}.
In addition, QWs in graph \cite{Aharonov-1}, on a line with a moving
boundary \cite{Kwek}, with multiple coins \cite{Brun-1} or decoherent
coins \cite{Brun} have been discussed also.

QWs using a SU(2) coin was introduced by Chandrashekar, \textsl{et
al}. \cite{Chandrashekar}, where the standard deviation and measurement
entropy properties were discussed. Here we discuss the symmetry and
average position properties for the QWs with a U(2) coin.

\section{hadamard quantum walks}

In this paper, we always discuss within the discrete-time QWs . The
total Hilbert space for QWs is given by $\mathcal{H}\equiv\mathcal{H}_{P}\otimes\mathcal{H}_{C}$,
where $\mathcal{H}_{P}$ is spanned by the orthonormal states $\left\{ \mid x\rangle\right\} $
and $\mathcal{H}_{C}$ is the two-dimensional coin space spaned by
two orthonormal states $\mid L\rangle$ and $\mid R\rangle.$

Each step of the QWs can be splitted into two operations: the evolution
of coin state and the particle movement according to the coin state.

Here the Hadamard walk, the coin is evolved by applying the Hadamard
operation:

\[
H=\frac{1}{\sqrt{2}}\begin{pmatrix}1 & 1\\
1 & -1
\end{pmatrix},
\]
the particle movement operator is given by

\begin{equation}
S=e^{ip\sigma_{z}}=\sum_{x}S_{x},
\end{equation}
where $p$ is the momentum operator, $\sigma_{z}$ is the Pauli-$z$
operator, 

\begin{equation}
\hat{S}_{x}=\mid x+1\rangle\langle x\mid\otimes\mid R\rangle\langle R\mid+\mid x-1\rangle\langle x\mid\otimes\mid L\rangle\langle L\mid.
\end{equation}
Therefore, after $t$ steps QWs 
\begin{equation}
\begin{aligned}\mid\Psi_{t}\rangle & =\left[S(I_{P}\otimes H)\right]^{t}\mid\Psi_{0}\rangle\\
 & =\left[\sum_{x}S_{x}(I_{P}\otimes H_{C})\right]^{t}\mid\Psi_{0}\rangle,
\end{aligned}
\label{eq:hadamardDQW}
\end{equation}
where $\mid\Psi_{in}\rangle$ is the initial state of the system.

\section{generalized discrete time quantumw walks}

An arbitrary one-qubit unitary operation can be written as a U(2)
matrix: 

\begin{equation}
U_{\alpha,\beta,\gamma,\theta}=e^{i\theta}\left(\begin{array}{cc}
e^{i\alpha}\cos\beta, & -e^{-i\gamma}\sin\beta\\
e^{i\gamma}\sin\beta, & e^{-i\alpha}\cos\beta
\end{array}\right)\label{eq:U(2)}
\end{equation}
For example, the Hadamard operator can be described in the form $H=U_{\frac{\pi}{2},\frac{\pi}{4},\frac{\pi}{2},-\frac{\pi}{2}}$.
By replacing the Hadamard coin with an operator $U_{\alpha,\beta,\gamma,\theta}$,
we can obtain the generalized QWs \cite{Chandrashekar}, which can
be written as 

\begin{equation}
\mid\Psi_{t}\rangle=\left[S(I_{P}\otimes U_{\alpha,\beta,\gamma,\theta})\right]^{t}\mid\Psi_{0}\rangle
\end{equation}

\begin{lem}
\label{theo:u2-su2}The quantum walks have the same probability distribution
with a U(2) coin or a SU(2) coin which has the same parameters $\alpha$,
$\beta$ and $\gamma$ in the U(2) martrix.\end{lem}
\begin{proof}
The SU(2) coin operator can be written as 
\begin{equation}
U_{\alpha,\beta,\gamma}^{S}=\left(\begin{array}{cc}
e^{i\alpha}\cos\beta, & -e^{-i\gamma}\sin\beta\\
e^{i\gamma}\sin\beta, & e^{-i\alpha}\cos\beta
\end{array}\right),\label{eq:SU(2)}
\end{equation}
Then the U(2) coin $U_{\alpha,\beta,\gamma,\theta}=e^{i\theta}U_{\alpha,\beta,\gamma}^{S}$.
The probability distribution for QWs with a U(2) coin after $t$ steps:
\begin{equation}
\begin{aligned}P(x,t) & =|\langle x|\Psi_{t}\rangle|^{2}=|e^{it\theta}\langle x\mid\Psi_{t}^{S}\rangle|^{2}\\
 & \equiv P^{S}(x,t),
\end{aligned}
\end{equation}
where $\mid\Psi_{t}^{S}\rangle$ and $P^{S}(x,t)$ are the state and
the probability distribution for QWs with a SU(2) coin after $t$
steps respectively.\end{proof}
\begin{cor}
\label{coro:position-average}The average position is the same in
quantum walks with a U(2) coin and a SU(2) coin if the two coins have
the same parameters $\alpha$, $\beta$ and $\gamma$ with the U(2)
coin.\end{cor}
\begin{proof}
From Lemma \ref{theo:u2-su2}, we can know that the average position
for a U(2) coin $\langle x\rangle=\sum_{x}xP(x,t)=\sum_{x}xP^{S}(x,t)\equiv\langle x\rangle^{S}$,
where $\langle x\rangle^{S}$ is the average position for QWs with
a SU(2) coin.
\end{proof}
With corollary \ref{coro:position-average}, if we want to know the
average position property in QWs with an arbitray unitary operator,
we only need to study the quantum walk with a SU(2) coin instead,
for the rest part of this paper, we always use the SU(2) coin as denoted
in Eq. \eqref{eq:SU(2)}.

Following the analysis in Ref. \cite{Nayak}, the state after $t$
steps QWs with a SU(2) coin

\begin{equation}
\mid\Psi_{t}\rangle=\left[S(I_{P}\otimes U_{\alpha,\beta,\gamma}^{S})\right]^{t}\mid\Psi_{0}\rangle,
\end{equation}
 the spatial Fourier transformation for of the wave function $\Psi(x,t)$
over $\mathscr{\mathcal{Z}}$ is given by

\begin{equation}
\begin{aligned}\tilde{\Psi}(k,t) & =\sum_{x=-\infty}^{\infty}\Psi(x,t)e^{ikx}\end{aligned}
,
\end{equation}
where $k\in[-\pi,\pi]$. We can know

\begin{equation}
\tilde{\Psi}(k,t)=(M_{k})^{t}\tilde{\Psi}(k,0),
\end{equation}
where 
\begin{equation}
\begin{aligned}M_{k} & =e^{ik}M_{+}+e^{-ik}M_{-}\\
 & =\left(\begin{array}{cc}
e^{-i(k-\alpha)}\cos\beta, & -e^{-i(k+\gamma)}\sin\beta\\
e^{i(k+\gamma)}\sin\beta, & e^{i(k-\alpha)}\cos\beta
\end{array}\right).
\end{aligned}
\end{equation}
The eigenvalues of $M_{k}$ is 
\begin{equation}
\begin{cases}
\lambda_{a}=e^{-iw}\\
\lambda_{b}=e^{iw}
\end{cases},
\end{equation}
where $\cos w=\cos(k-\alpha)\cos\beta$. And the eigenstates 
\begin{equation}
\begin{cases}
\tilde{\Psi}_{k}^{a} & =\frac{1}{C_{k}^{a}}\left(\begin{array}{c}
P_{k}\\
Q_{k}^{a}
\end{array}\right)\\
\tilde{\Psi}_{k}^{b} & =\frac{1}{C_{k}^{b}}\left(\begin{array}{c}
P_{k}\\
Q_{k}^{b}
\end{array}\right)
\end{cases},
\end{equation}
where
\begin{equation}
\begin{aligned}\begin{cases}
P_{k}=-e^{-i(k+\gamma)}\sin\beta\\
Q_{k}^{a}=-i\sin\omega_{k}+i\sin(k-\alpha)\cos\beta\\
Q_{k}^{b}=i\sin\omega_{k}+i\sin(k-\alpha)\cos\beta
\end{cases}\end{aligned}
,
\end{equation}
\begin{equation}
\begin{aligned}C_{k}^{a} & =\sqrt{P_{k}^{*}P_{k}+(Q_{k}^{a})^{*}Q_{k}^{a}}\\
 & =\sqrt{2(\sin^{2}\omega_{k}-\cos\beta\sin(k-\alpha)\sin\omega_{k})},
\end{aligned}
\end{equation}
\begin{equation}
\begin{aligned}C_{k}^{b} & =\sqrt{P_{k}^{*}P_{k}+(Q_{k}^{b})^{*}Q_{k}^{b}}\\
 & =\sqrt{2(\sin^{2}\omega_{k}+\cos\beta\sin(k-\alpha)\sin\omega_{k})}.
\end{aligned}
\end{equation}

\begin{figure}
\includegraphics[width=3in]{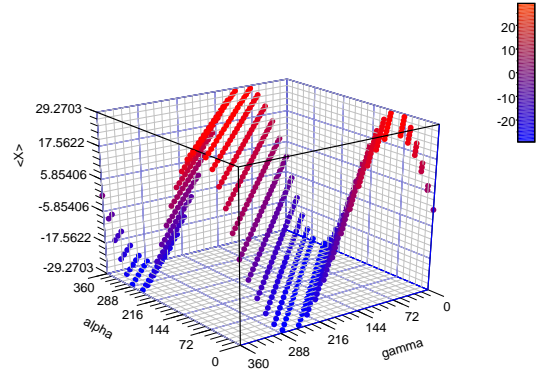}

\caption{The average of position $\left\langle x\right\rangle $ for quantun
walks after $t=100$ steps with a SU(2) coin, where$\beta=\pi/6$.}
\label{fig:ave}
\end{figure}

\begin{figure}
\includegraphics[width=3in]{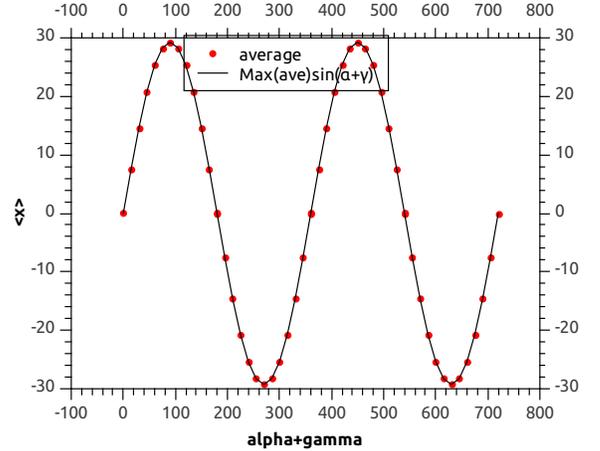}

\caption{(red dot)The average of position with $\alpha+\gamma$ after $100$
steps quantum walk when $\beta=\pi/6$, (line) $f(\phi)=\max(\left\langle x\right\rangle )\sin(\phi)$. }
\label{fig:2}
\end{figure}

\section{the average position in quantum walks}

Fig. \ref{fig:ave} and Fig. \ref{fig:2} show the average position
after $t=100$ steps QWs with a SU(2) coin in the case of $\beta=\pi/6$,
while the initial state is $1/\sqrt{2}(\mid0L\rangle+i\mid0R\rangle)$.
From Fig. \ref{fig:ave}, we can know that $\left\langle x\right\rangle $
only depends on the sum of $\alpha$ and $\gamma$. Fig. \ref{fig:2}
shows that the actual $\left\langle x\right\rangle $ exactly match
the function of $f(\phi)=\max(\left\langle x\right\rangle )\sin(\phi)$,
then we conject that $\left\langle x\right\rangle =G(\beta,t)\sin(\alpha+\gamma)$.

\section{proof in methematics}
\begin{thm}
The probability distribution for quantum walks with a SU(2) coin is
independent on the parameter $\alpha$ and $\gamma$, when the initial
state is $\mid0L\rangle$ or $\mid0R\rangle.$ i.e.
\begin{equation}
\begin{aligned}\begin{cases}
P_{\mid0L\rangle}(\alpha,\beta,\gamma,x,t)=P_{\mid0L\rangle}(\beta,x,t)\\
P_{\mid0R\rangle}(\alpha,\beta,\gamma,x,t)=P_{\mid0R\rangle}(\beta,x,t)
\end{cases}\end{aligned}
,
\end{equation}
for any $\alpha$ and $\gamma.$\end{thm}
\begin{proof}
If the initial state $\mid\Psi_{0}\rangle=\mid0L\rangle$, then $\tilde{\Psi}(k,0)=\left(\begin{array}{c}
1\\
0
\end{array}\right)$. The probability of $\mid x\rangle:$

\begin{equation}
\begin{aligned} & P_{\mid0L\rangle}(\alpha,\beta,\gamma,x,t)\\
 & =P_{L}(\alpha,\beta,\gamma,x,t)+P_{R}(\alpha,\beta,\gamma,x,t)\\
 & =\frac{1}{4\pi^{2}}\int_{-\pi}^{\pi}\int_{-\pi}^{\pi}g_{L}(k_{1},k_{2},\alpha,\beta,\gamma,x,t)dk_{1}dk_{2}+\\
 & \qquad\frac{1}{4\pi^{2}}\int_{-\pi}^{\pi}\int_{-\pi}^{\pi}g_{R}(k_{1},k_{2},\alpha,\beta,\gamma,x,t)dk_{1}dk_{2},
\end{aligned}
\end{equation}
where $g_{j}(k_{1},k_{2},\alpha,\beta,\gamma,x,t)=\tilde{\Psi}_{j}^{*}(k_{1},t)\tilde{\Psi}_{j}(k_{2},t)e^{i(k_{1}-k_{2})x}$,
$j\in\{L,R\}$. If we set $h=k-\alpha$, then we can know $g_{j}(k_{1},k_{2},\alpha,\beta,\gamma,x,t)=g_{j}(h_{1},h_{2},\beta,x,t)$,
and $g_{j}(h_{1},h_{2},\beta,x,t)=g_{j}(h_{1}\pm2\pi,h_{2},\beta,x,t)=g_{j}(h_{1},h_{2}\pm2\pi,\beta,x,t)$,
then 

\begin{equation}
\begin{aligned}P_{j}(\alpha,\beta,\gamma,x,t) & =\frac{1}{4\pi^{2}}\int_{-\pi-\alpha}^{\pi-\alpha}\int_{-\pi-\alpha}^{\pi-\alpha}g_{j}(h_{1},h_{2},\beta,x,t)dh_{1}dh_{2}\\
 & =\frac{1}{4\pi^{2}}\int_{-\pi}^{\pi}\int_{-\pi}^{\pi}g_{j}(h_{1},h_{2},\beta,x,t)dh_{1}dh_{2}\\
 & =P_{j}(\beta,x,t).
\end{aligned}
\end{equation}
Further more we can know $P_{\mid0L\rangle}(\alpha,\beta,\gamma,x,t)=P_{\mid0L\rangle}(\beta,x,t)$.
In the same way, we can also get $P_{\mid0R\rangle}(\alpha,\beta,\gamma,x,t)=P_{\mid0R\rangle}(\beta,x,t)$. \end{proof}
\begin{thm}
\label{theo:0L-0R}After t steps, we have
\begin{equation}
\begin{cases}
\Psi_{\mid0L\rangle}^{R}(x,t)+\Psi_{\mid0R\rangle}^{L}(-x,t)\in i\mathbb{R}\\
\Psi_{\mid0L\rangle}^{R}(x,t)-\Psi_{\mid0R\rangle}^{L}(-x,t)\in\mathbb{R}
\end{cases},
\end{equation}
\begin{equation}
\begin{cases}
\Psi_{\mid0L\rangle}^{L}(x,t)+\Psi_{\mid0R\rangle}^{R}(-x,t)\in\mathbb{R}\\
\Psi_{\mid0L\rangle}^{L}(x,t)-\Psi_{\mid0R\rangle}^{R}(-x,t)\in i\mathbb{R}
\end{cases},\label{eq:lr}
\end{equation}
where $\Psi_{\mid n\rangle}^{m}(x,t)$ denotes the coefficient of
the $\mid x\rangle\mid m\rangle$ state after $t$ steps quantum walk
with the initail state $\mid n\rangle$.\end{thm}
\begin{proof}
\begin{equation}
\begin{aligned} & \Psi_{\mid0L\rangle}^{R}(x,t)\pm\Psi_{\mid0R\rangle}^{L}(-x,t)\\
 & =\frac{1}{2\pi}\int_{-\pi}^{\pi}\tilde{\Psi}_{\mid0L\rangle}^{R}(k,t)e^{-ikx}dk\pm\frac{1}{2\pi}\int_{-\pi}^{\pi}\tilde{\Psi}_{\mid0R\rangle}^{L}(k,t)e^{ikx}dk\\
 & =\frac{1}{2\pi}\int_{-\pi}^{\pi}\frac{Q_{k}^{a}}{(C_{k}^{a})^{2}}[P_{k}^{*}e^{-ikx}\mp P_{k}e^{ikx}][e^{-i\omega_{k}t}-e^{i\omega_{k}t}]dk.
\end{aligned}
\end{equation}
As $Q_{k}^{a}\in i\mathbb{R}$, we can know $\Psi_{\mid0L\rangle}^{R}(x,t)+\Psi_{\mid0R\rangle}^{L}(-x,t)\in i\mathbb{R}$,
and $\Psi_{\mid0L\rangle}^{R}(x,t)-\Psi_{\mid0R\rangle}^{L}(-x,t)\in\mathbb{R}.$
Similarly, we can get Eq. \eqref{eq:lr}.\end{proof}
\begin{cor}
\label{coro:symmetry-LR}The symmetry property of distribution between
quantum walks with a U(2) coin in initial state $\mid0L\rangle$ and
$\mid0R\rangle$: For an arbitrary $t$, $P_{\mid0L\rangle}^{R}(\beta,x,t)=P_{\mid0R\rangle}^{L}(\beta,-x,t)$,
$P_{\mid0L\rangle}^{L}(\beta,x,t)=P_{\mid0R\rangle}^{R}(\beta,-x,t)$.\end{cor}
\begin{proof}
We set $\Psi_{\mid0L\rangle}^{R}(x)=C+Di$, where $C,D\in\mathbb{R}$.
From Theorem \ref{theo:0L-0R}, we can know $\Psi_{\mid0R\rangle}^{L}(-x)=-C+Di$,
then we can know $P_{\mid0L\rangle}^{R}(\beta,x,t)=P_{\mid0R\rangle}^{L}(\beta,-x,t)$.
Similarly, we can also get $P_{\mid0L\rangle}^{L}(\beta,x,t)=P_{\mid0R\rangle}^{R}(\beta,-x,t)$.\end{proof}
\begin{thm}
\label{theo:mn}If the initial state $\mid\Psi_{0}\rangle=m\mid0L\rangle+n\mid0R\rangle$,
where $|m|^{2}+|n|^{2}=1$, the probability at state $\mid xL\rangle$
or $\mid xR\rangle$ after $t$ steps quantum walk is
\begin{widetext}
\begin{equation}
\begin{cases}
P^{L}(x)=|m|^{2}P_{\mid0L\rangle}^{L}+|n|^{2}P_{\mid0R\rangle}^{L}-(e^{-i(\alpha+\gamma)}m^{*}n+e^{i(\alpha+\gamma)}mn^{*})G^{L}(\beta,x,t)\\
P^{R}(x)=|m|^{2}P_{\mid0L\rangle}^{R}+|n|^{2}P_{\mid0R\rangle}^{R}-(e^{-i(\alpha+\gamma)}m^{*}n+e^{i(\alpha+\gamma)}mn^{*})G^{R}(\beta,x,t)
\end{cases},\label{eq:PLMN0}
\end{equation}

\end{widetext}
where $G^{L}$ and $G^{R}$ are indepent of $\alpha$ and $\gamma$.\end{thm}
\begin{proof}
The probability at state $\mid xL\rangle$ after $t$ steps:
\begin{widetext}
\begin{equation}
\begin{aligned}P_{m\mid0L\rangle+n\mid0R\rangle}^{L}(x) & =\frac{1}{4\pi^{2}}\int\int\tilde{\Psi}_{L}^{*}(k_{1},t)\tilde{\Psi}_{L}(k_{2},t)e^{i(k_{1}-k_{2})x}dk_{1}dk_{2}\\
 & =|m|^{2}P_{0L}^{L}(\beta,x,t)+|n|^{2}P_{0R}^{L}(\beta,x,t)-\\
 & \qquad(e^{-i(\alpha+\gamma)}m^{*}n\sum_{i=1}^{4}G_{i}+e^{i(\alpha+\gamma)}mn^{*}\sum_{i=1}^{4}G_{i}^{*}),
\end{aligned}
\label{eq:PLMN}
\end{equation}
where

\begin{equation}
\begin{aligned}G_{1}(\beta,x,t) & =\frac{1}{4\pi^{2}}\int\int e^{i(\omega_{h_{1}}-\omega_{h_{2}})t}\frac{1}{(C_{h_{1}}^{b}C_{h_{2}}^{b})^{2}}e^{i(h_{1}-h_{2})}\sin^{3}\beta e^{i(h_{1}-h_{2})x}e^{-ih_{1}}(Q_{h_{2}}^{b})^{*}dh_{1}dh_{2}\\
 & =\frac{1}{4\pi^{2}}\int\int\frac{(Q_{h_{2}}^{b})^{*}\sin^{3}\beta}{(C_{h_{1}}^{b}C_{h_{2}}^{b})^{2}}i\sin[(\omega_{h_{1}}-\omega_{h_{2}})+(h_{1}-h_{2})+(h_{1}-h_{2})x-h_{1}]dh_{1}dh_{2}\in R.
\end{aligned}
\end{equation}

As the same of $G_{1}(\beta,x,t)$, we can know $G_{i}(\beta,x,t)\in R$,
where $i\in\{1,2,3,4\}$. So Eq. \eqref{eq:PLMN} can be written as
\begin{equation}
P^{L}(x)=|m|^{2}P_{\mid0L\rangle}^{L}+|n|^{2}P_{\mid0R\rangle}^{L}-(e^{-i(\alpha+\gamma)}m^{*}n+e^{i(\alpha+\gamma)}mn^{*})G^{L}(\beta,x,t),
\end{equation}

\end{widetext}
where $G^{L}(\beta,x,t)=\sum_{i=1}^{4}G_{i}$. In the same way as
$P^{L}(x)$, we can get $P^{R}(x)$ in Eq. \eqref{eq:PLMN0}. \end{proof}
\begin{thm}
If the initail state $\mid\Psi_{0}\rangle=1/\sqrt{2}(\mid0L\rangle+i\mid0R\rangle)$,
The average position after $t$ steps quantum walk: $\left\langle x\right\rangle =G(\beta,t)\sin(\alpha+\gamma)$,
where $G(\beta,t)$ only depends on $\beta$ and $t$.\end{thm}
\begin{proof}
From Corollary \ref{coro:symmetry-LR}, we can know
\begin{equation}
\begin{cases}
\sum_{x=-t}^{t}x[P_{\mid0L\rangle}^{R}(\beta,x,t)+P_{\mid0R\rangle}^{L}(\beta,x,t)]=0\\
\sum_{x=-t}^{t}x[P_{\mid0L\rangle}^{L}(\beta,x,t)+P_{\mid0R\rangle}^{R}(\beta,x,t)]=0
\end{cases}\text{.}\label{eq:PLR0}
\end{equation}
 Using Eq. \eqref{eq:PLR0} and Theorem \ref{theo:mn} we can know
\begin{equation}
\begin{aligned}\left\langle x\right\rangle  & =\sum_{x=-t}^{t}x(P^{L}(\alpha,\beta,\gamma,x,t)+P^{R}(\alpha,\beta,\gamma,x,t))\\
 & =G(\beta,t)\sin(\alpha+\gamma),
\end{aligned}
\end{equation}
where $G(\beta,t)=-\sum_{x=-t}^{t}x[G^{L}(\beta,x,t)+G^{R}(\beta,x,t)]$
only depends on $\beta$ and $t$, regardless of $\alpha$ or $\gamma$. 
\end{proof}

\section{conclusions\label{sec:conclusions}}

In this paper, we discussed the properties of the average position
in QWs with an arbitrary unitary coin. With a SU(2) coin, if the initial
state is $\mid0L\rangle$ or $\mid0R\rangle$, the probability distribution
is independent on $\alpha$ and $\gamma$. Some symmetry properties
between different initial states $\mid0L\rangle$ and $\mid0R\rangle$
was prooved, we get that $P_{\mid0L\rangle}^{R}(\beta,x,t)=P_{\mid0R\rangle}^{L}(\beta,-x,t)$
and $P_{\mid0L\rangle}^{L}(\beta,x,t)=P_{\mid0R\rangle}^{R}(\beta,-x,t)$.
If the initial state $\mid\Psi_{0}\rangle=1/\sqrt{2}(\mid0L\rangle+i\mid0R\rangle)$,
we can know the average $\left\langle x\right\rangle =G(\beta,t)\sin(\alpha+\gamma)$,
so if we replace the Hadamard operator with an arbitrary unitary operator,
the average position is always not equal to $0$, unless $\alpha+\beta=n\pi,n\in\mathbb{Z}$.
\begin{acknowledgments}
This work was supported by the National Natural Science Foundation
of China (Grant No. 10974192, 61275122), the National Fundamental
Research Program of China (Grant No. 2011CB921200, 2011CBA00200),
K. C. Wong Education Foundation and CAS.\end{acknowledgments}

\end{document}